\documentclass[11pt]{article}

\usepackage{lmodern}
\usepackage[T1]{fontenc}
\usepackage[utf8]{inputenc}
\usepackage{fullpage}
\usepackage{hyperref}
\usepackage{amsmath, amssymb, amsthm}
\usepackage{bbm}
\usepackage{mathtools}
\usepackage{verbatim}
\usepackage{color}
\usepackage{xcolor}
\usepackage{authblk}
\usepackage{multicol}
\usepackage{tikz}
\usetikzlibrary{math, calc}

\definecolor{DarkGreen}{rgb}{0.1,0.5,0.1}
\definecolor{DarkRed}{rgb}{0.5,0.1,0.1}
\definecolor{DarkBlue}{rgb}{0.1,0.1,0.5}
\usepackage{xspace}
\usepackage[numbers,sort&compress]{natbib}
\usepackage{cleveref}

\usepackage[mathscr]{euscript}
 \let\mathscr\relax
\usepackage[scr]{rsfso}

\newcommand{\poa}{\mathrm{PoA}}
\newcommand{\pos}{\mathrm{PoS}}

\newcommand{\opt}{\ensuremath{\mathrm{OPT}}}

\renewcommand{\tilde}{\widetilde}

\newcommand{\ex}[1]{\mathbb{E}\left[#1\right]}

\newcommand{\mechav}{\ensuremath{\mathcal{M}_{\mathrm{AV}}}}

\def\epsilon{\varepsilon}

\DeclareMathOperator*{\argmax}{\mathrm{argmax}}

\theoremstyle{plain}
\newtheorem{theorem}{Theorem}[section]
\newtheorem{corollary}[theorem]{Corollary}

\newtheorem{proposition}[theorem]{Proposition}

\theoremstyle{definition}
\newtheorem{definition}[theorem]{Definition}

\theoremstyle{remark}
\newtheorem*{remark}{Remark}

\title{Decentralised Update Selection with Semi-Strategic Experts\thanks{\,This work was supported by the ERC Advanced 
 Grant 788893 AMDROMA ``Algorithmic and Mechanism Design Research in 
 Online Markets'', the MIUR PRIN project ALGADIMAR ``Algorithms, Games, 
 and Digital Markets'', and the NWO Veni project No.~VI.Veni.192.153.}}
\author[1]{Georgios Amanatidis}
\author[2]{Georgios Birmpas}
\author[4]{Philip Lazos}
\author[3,4]{Francisco Marmolejo-Cossío}

\affil[1]{\small Department of Mathematical Sciences; University of Essex}
\affil[ ]{{\small\textsf{\href{mailto:georgios.amanatidis@essex.ac.uk}{georgios.amanatidis}@essex.ac.uk}}\smallskip}

\affil[2]{\small Department of Computer, Control and Management Engineering; Sapienza University of Rome}
\affil[ ]{{\small\textsf{\href{mailto:birbas@diag.uniroma1.it}{birbas@diag.uniroma1.it}}}\smallskip}

\affil[3]{{\small School of Engineering and Applied Sciences; Harvard University}}
\affil[ ]{{\small\textsf{\href{mailto:fjmarmol@seas.harvard.edu}{fjmarmol@seas.harvard.edu}}}\smallskip}

\affil[4]{{\small IOHK}}
\affil[ ]{{\small\textsf{
\{\href{mailto:philip.lazos@iohk.io}{philip.lazos}, 
\href{mailto:francisco.marmolejo@iohk.io}{francisco.marmolejo}\}@iohk.io}}\smallskip}

\date{\today}

\begin{document}

\maketitle

\begin{abstract}
     Motivated by governance models adopted in blockchain applications, we study the problem of selecting appropriate system updates in a decentralised way. Contrary to most existing voting approaches, we use the input of a set of motivated experts of varying levels of expertise. In particular, we develop an approval voting inspired selection mechanism through which the experts approve or disapprove the different updates according to their perception of the quality of each alternative. Given their opinions, and weighted by their expertise level, a single update is then implemented and evaluated, and the experts receive rewards based on their choices. We show that this mechanism always has approximate pure Nash equilibria and that these achieve a constant factor approximation with respect to the quality benchmark of the optimal alternative. Finally, we study the repeated version of the problem, where the weights of the experts are adjusted after each update, according to their performance. Under mild assumptions about the weights, the extension of our mechanism still has approximate pure Nash equilibria in this setting.
\end{abstract}

\newcommand{\vecc}[1]{\ensuremath{\mathbf{#1}}}
\newcommand{\mech}{\ensuremath{\mathcal{M}}}
\newcommand{\sset}[1]{\left\{ #1\right\}}
\newcommand{\ssets}[1]{\{ #1\}}
\newcommand{\fwh}[1]{\; \left| \; #1 \right.}
\newcommand{\fwhs}[1]{\; | \; #1 }
\newcommand{\winning}{\ensuremath{j^\star}}
\allowdisplaybreaks

\newcommand{\repcoin}{\texttt{ReputationCoin}}
\newcommand{\RCcoin}{\ensuremath{\textsc{RC}}}
\newcommand{\expertP}{\ensuremath{P}}
\newcommand{\meritcoin}{\textttt{MeritCoin}}

\newcommand{\nExperts}{\ensuremath{n_e}}
\newcommand{\nVoters}{\ensuremath{n_v}}
\newcommand{\nUpdates}{\ensuremath{n_u}}

\newcommand{\trueP}{\ensuremath{\texttt{TP}}}
\newcommand{\trueN}{\ensuremath{\texttt{TN}}}
\newcommand{\gainE}{\ensuremath{\texttt{ge}}}
\newcommand{\gainV}{\ensuremath{\texttt{gv}}}

\section{Introduction}

In 2009, Satoshi Nakamoto published a landmark whitepaper outlining the core functionality of Bitcoin \cite{nakamoto2008bitcoin}, a decentralized blockchain-based ledger of transactions for a peer-to-peer digital currency. Indeed, a key feature of Bitcoin is precisely in its decentralized nature, whereby no single entity controls the operation of the system, a feat which is achieved by an innovative amalgamation of cryptography and carefully aligned incentives amongst participants in the protocol. In recent years the ecosystem for similar decentralized systems has grown drastically, and given that for each of these systems no single entity holds control, users often find themselves in a position where they need to reach a consensus on critical decisions regarding the very platform they participate in. 

A fundamental example of this governance dilemma is that of creating and implementing software updates for the underlying infrastructure of a blockchain-based solution \cite{vitalik3}. 
Such drastic upgrades are known as hard forks in blockchain-based systems, and historically, there have been scenarios where cryptocurrencies have split due to opinion differences regarding the infrastructure of the blockchain (e.g., Ethereum vs.~Ethereum Classic or Bitcoin vs.~Bitcoin Cash). Beyond the confusion and inconvenience for the users caused by such splits, these fragmentations have very real implications for the security of blockchain-based systems, which is often strongly dependant on the number of users within a system.

Software updates, as opposed to general collective decisions, are particularly interesting due to two salient features of the problem structure: 1) adequately evaluating the relative merits of software update proposals often requires a high degree of expertise 2) the overall stakes of the software update process are incredibly high. Indeed, if a proposal which is collectively chosen for uptake happens to have a fatal bug which has failed to be caught, its uptake can have catastrophic implications for the underlying system. 

In this work, we focus on providing a simple on-chain methodology whereby users from a blockchain-based platform can collectively decide which software updates to implement, thereby reducing the potential for aforementioned hard fork frictions. To do so, we assume the existence of a set of users with differing levels of expertise regarding software updates. These experts are then faced with proposals for multiple potential software updates, where each expert not only formulates independent opinions on whether a given proposal will succeed or not during implementation, but may also harbor additional incentives for implementing one proposal over another. Each expert casts votes in approval or disapproval of each proposal, and as a function of these votes and the historical merit of each expert, a proposal is chosen for implementation. Ultimately, we assume the success or failure of the proposal is recorded on-chain, and as a function of each expert's votes and the outcome of the proposal, the system pays experts for their participation.

\subsection{Our Contributions}

Our main contributions can be summarized as follows:
\begin{itemize}
    \item In Section \ref{sec:preliminaries} we introduce our framework for encoding expert preferences amongst software proposals, the notion of semi-strategic decision making amongst experts, and a benchmark for proposal quality which we subsequently use to measure the performance of equilibria in the semi-strategic voting setting.
    \item In Section \ref{sec:approval-voting}, we present our approval voting inspired mechanism. We provide some sufficient conditions that limit potentially damaging deviations and show existence and Price of Anarchy (PoA) results for pure Nash equilibria. In Sections \ref{sec:approx-equilibria} and \ref{sec:PoA}, we study approximate pure Nash equilibria for semi-strategic experts and show that the PoA in this setting is exactly 2.
    \item Finally, in Section \ref{sec:repeated} we consider a repeated game setting that dynamically reflects expert performance via weight updates in rounds of approval voting (a measure of ``reputation''). In this setting we show honest voting is an approximate pure Nash equilibrium and that repeated voting has a Price of Stability of at most 2. 
    
\end{itemize}


\subsection{Related Work}

Due to the fact that our work focuses on creating payment mechanisms whereby a blockchain platform can elicit truthful expert opinions regarding potential software updates, it bears many similarities to existing literature in the realm of scoring rules and preference elicitation \cite{brier1950verification,savage1971elicitation,gneiting2007strictly,lambert2009eliciting}. 
In the context of software updates, implementing a classic scoring rule for one proposal would involve asking experts to report a probability $p \in [0,1]$ corresponding to their belief that a proposal would succeed or fail. 
As a function of this reported probability and the outcome of the proposal, the scoring rule pays experts in such a way that they always maximize their expected payment when reporting truthfully. 
The works of \citet{chen2014eliciting}, \citet{chen2011information}, and \citet{othman2010decision}, extend the scoring rule framework to incorporate decision-making between mutually exclusive alternatives, as is the case with our proposal setting. Indeed only one proposal is chosen, against which we must score expert performance. 

Our work however also considers the very real possibility (especially in the semi-anonymous and permissionless world of blockchain platforms) that expert incentives go beyond the payment which the mechanism offers them for providing their opinions, and that they instead exhibit distinct utilities for the ultimate implementation of different proposals. 
This assumption is similar to the setting of Proper Decision Scoring Rules, as explored in the work of \citet{oesterheld2020decision}, as well as Restricted Investment Contracts presented in the paper of \citet{carroll2019robust}, where potential solutions involve experts earning rewards proportional to the principal's earnings for the alternative chosen and they each have potentially conflicting beliefs over the quality of the various proposals. 

The mechanism we propose for eliciting expert beliefs builds off of a rich existing literature \cite{RePEc:spr:stchwe:978-3-642-02839-7} on Approval Voting (AV) first introduced by \citet{brams1978approval}, and then extended in \cite{fishburn1978strategic,fishburn1981approval,endriss2013sincerity}.
As in AV, experts voting consists of providing a subset of proposals which they approve (and, thus, a subset which they disapprove). Unlike typical AV, the underlying utility that experts derive from the outcome of the mechanism is tied to their beliefs regarding the success of each proposal. Furthermore, as mentioned in the work of \citet{laslier2010basic}, a common technique in AV involves restricting voter actions in natural ways (i.e., admissible and sincere voting profiles from \cite{fishburn1981approval} for example) to restrict the set of equilibria considered in the mathematical analysis of the model. Our work introduces a novel constraint of similar nature which we dub semi-strategic voting (see Section \ref{sec:preliminaries}).  For semi-strategic experts we demonstrate theoretical guarantees on our mechanism's equilibria resulting from AV. Finally, it is worth mentioning that AV has additionally seen much empirical work pointing to its practical success \cite{brams1991approval,regenwetter1998approval,brams2010going}.

\section{Preliminaries}
\label{sec:preliminaries}

    Let $N = \{1, 2, \ldots, n\}$ be a set of experts (typically indexed by $i$) and $U = \{1, 2, \ldots, k\}$ a set of proposed updates (typically indexed by $j$). Every update $j$ has an associated quality $q_j \in \{0, 1\}$ indicating whether it is beneficial for the system; we call updates with $q_j = 1$ ``good'' and those with $q_j = 0$ ``bad''. These $q_j$ values \emph{are not initially known}: only the value of the selected proposal is revealed \emph{after} the experts' vote has concluded. Specifically, each expert has her own prior about the $q_j$'s. Expert $i$ believes that $q_j$ has a probability $p_{ij}$ of being good. This `opinion' $p_{ij}$ is unaffected by what the other experts think or how they voted: no matter how some update $j$ was selected, expert $i$ believes that it will be good with probability $p_{ij}$. Every expert $i$ has a weight $w_i \ge 0$ indicating their ``power'' within the system. Additionally, each expert may have some (external) personal gain, depending on the outcome of the vote. We denote this external reward that expert $i$ will receive if update $j$ is implemented by $g_{ij}$; this value is known to expert $i$, but not to the mechanism.

The strategy of each expert is a vector 
${\vecc r_i} = (r_{i1}, \allowbreak\ldots, \allowbreak r_{ik})^\intercal \in \{0,1\}^{k\times 1}$, where $r_{ij}$ indicates the binary vote of expert $i$ on whether $q_j = 1$ or not. 
Let ${\vecc r} = \allowbreak (\vecc r_1,\allowbreak \vecc r_2, \ldots, \vecc r_n) \in \{0,1\}^{k \times n}$ denote the whole \emph{voting profile}. As is common, we write $(\vecc r'_i, {\vecc r_{-i}})$ to denote $(\vecc r_1, \ldots, \vecc r_{i-1}, \vecc r'_i, \vecc r_{i+1}, \ldots,  \vecc r_{n})$, as well as $(r'_{ij}, {\vecc r_{-ij}})$ to denote $((r_{i1}, \allowbreak\ldots,r_{i (j-1)}, r'_{ij}, r_{i (j+1)} \allowbreak, \ldots, r_{ik})^\intercal, {\vecc r_{-i}})$.

A mechanism $\mech = (x, \vecc f)$ consists of a (possibly randomised) selection rule $x$, which given $\vecc r$ and $\vecc w = (w_1, \ldots, w_n)$ returns
$x(\vecc r, \vecc w) \in U\cup\{0\}$ (i.e., $x$ could return the \emph{dummy} proposal 0, if no proposal is selected), and an expert reward function
$\vecc f$, which given $\vecc r$, $\vecc w$, a winning proposal $j^\star$, and its quality $q_{j^\star}$ returns $\vecc f(\vecc r, \vecc w, j^\star, q_{j^\star}) = (f_1(\vecc r, \vecc w, j^\star, q_{j^\star}),\ldots,f_n(\vecc r, \vecc w, j^\star, q_{j^\star})) \in \mathbb{R}^{1 \times n}$.

Given $\mech$, the expected reward of expert $i$ conditioned on $i$'s own perspective $\vecc p = (p_{i1}, \ldots, p_{in})$ uses an estimate of $f_i(\vecc r, \vecc w, j^\star, q_{j^\star})$ which depends on the knowledge of $\vecc r, \vecc w$, and $p_{ij^\star}$, i.e., the probability that $j^\star$ is a good proposal according to $i$ only:
\[ \ex{f_i(\vecc r, \vecc w, j^\star, q_{j^\star}) \,|\, \vecc p} =  p_{ij^\star} f_i(\vecc r, \vecc w, j^\star, 1) + (1-p_{ij^\star})  f_i(\vecc r, \vecc w, j^\star, 0) \,.\]
Thus, the expected utility of $i$ conditioned on her perspective
is
\begin{equation}
    u_i^{\mech}(\vecc r\;|\; \vecc w, \vecc p) = \ex{p_{ij^\star} \cdot g_{i j^\star} + \ex{f_i(\vecc r, \vecc w, j^\star, q_{j^\star})\,|\, \vecc p}}  \,,
\end{equation}
where the outer expectation is over the winning proposal $j^\star = x(\vecc r, \vecc w)$.
We adopt this approach for the utilities because when an expert declares her preference, we assume she is agnostic about the beliefs of other experts for any proposal.

%

Given the subjective evaluation of the quality of each proposal, we need a way to \emph{aggregate} the opinions of all experts that combines robustness and explainability. To this end, we introduce a probability threshold $T \in [0,1]$ such that if expert $i$ has $p_{ij} \ge T$ for proposal $j$ then we consider $i$'s honest response to be to vote in favour of $j$; otherwise $i$'s honest response is to vote against $j$.

Our metric can be viewed as the weighted average of the probabilities \emph{after these have been rounded to $0$ or $1$} with respect to the threshold $T$, and it has an immediate meaning which is the voting power that considers a proposal to be good enough. Note that using the raw  probabilities $p_{ij}$ to define some measure of quality is a bit problematic: it is not reasonable to expect that experts would precisely and consistently report those, and possibly it would be considerably harder to communicate the resulting notion of ``quality'' to non-experts in the system.


\begin{definition}[Estimated Quality]
    Given a probability threshold $T$, the \emph{estimated quality} of proposal $i$ is the sum of weights of experts $i$ with $p_{ij} \ge T$. That is:
    \begin{equation}
        \texttt{Qual}[j] = \!\sum_{i \,:\, p_{ij} \ge T}\! w_i.
    \end{equation}
\end{definition}
\noindent
For convenience, we will refer to the optimal quality as:
\begin{equation}
    \opt(\vecc w, \vecc p) = \argmax_j \texttt{Qual}[j].
\end{equation}

We consider experts that are strategic and strive to maximize their utility. However, we also assume they are not malicious towards the system. That is, they only choose to lie when this results in a net \emph{increase} in their utility. If there is no strictly beneficial deviation, they remain honest. We call such experts semi-strategic.
\begin{definition}[Semi-strategic Experts]
    An expert $i$ is \emph{semi-strategic} if for every mechanism $\mech$ and strategy vector $\vecc r$:
    \begin{itemize}
        \item If $p_{ij} \ge T$ and $r_{ij} = 0$, then 
        \[
            u_i^{\mech}(1, \vecc r_{-ij}\;|\; \vecc w, \vecc p) 
            < u_i^{\mech}(\vecc r\;|\; \vecc w, \vecc p)\,.
        \]
        \item If $p_{ij} < T$ and $r_{ij} = 1$, then 
        \[
            u_i^{\mech}(0, \vecc r_{-ij}\;|\; \vecc w, \vecc p)
            < u_i^{\mech}(\vecc r\;|\; \vecc w, \vecc p)\,.
        \] 
    \end{itemize}
\end{definition}
The solution concept we use is the (multiplicatively) approximate pure Nash equilibrium. We use the multiplicative, rather than the additive, version of approximate pure Nash equilibria as we want our results to be mostly independent of scaling up or down the reward functions.\footnote{This is completely precise in the case where the external rewards $g_{ij}$ are all $0$, but it is still largely true whenever the rewards of the mechanism are large compared to external rewards.}
\begin{definition}[$(1+\epsilon)$-Pure Nash Equilibrium]
    For $\epsilon \ge 0$, a strategy profile $\vecc r$ is a multiplicatively  {\em $(1+\epsilon)$-approximate pure Nash equilibrium}, or simply a $(1+\epsilon)$-PNE, for weight vector $\vecc w$, if for every deviation $\vecc r_i'$ we have:
    \begin{equation}
        (1+\epsilon)\cdot u_i^{\mech}(\vecc r \;|\; \vecc w, \vecc p) \ge u_i^{\mech}(\vecc r_i', \vecc r_{-i} \;|\; \vecc w, \vecc p)\,.
    \end{equation}
    When $\epsilon = 0$ we simply call $\vecc r$ a  {\em pure Nash equilibrium} (PNE).
\end{definition}
\noindent
We refer to the set strategies that are $(1+\epsilon)$-PNE of mechanism $\mech$ given $\vecc w$ and $\vecc p$ as $\mathcal{Q}_{\epsilon}^{\mech}(\vecc w, \vecc p)$.\medskip

To measure the inefficiency of different equilibria compared to the proposal of highest quality, we use the notions \emph{Price of Anarchy} \cite{KP09} and \emph{Price of Stability} \cite{AnshelevichDKTWR08}, which denote the ratios between the quality of the worst or the best possible equilibrium produced by $\mech = (x, \vecc f)$ and the optimal outcome, respectively. In particular, these are formally defined as:
$$
\poa(\mech)=\sup_{\vecc w, \vecc p}\frac{\opt(\vecc w, \vecc p)}{\inf_{\vecc r\in\mathcal{Q}^\mech_\epsilon(\vecc w, \vecc p)} \texttt{Qual}[x(\vecc r)]}$$
and
$$
\pos(\mech)=\sup_{\vecc w, \vecc p}\frac{\opt(\vecc w, \vecc p)}{\sup_{\vecc r\in\mathcal{Q}^\mech_\epsilon(\vecc w, \vecc p)} \texttt{Qual}[x(\vecc r)]}.
$$

\section{Approval Voting}
\label{sec:approval-voting}
Although our definitions allow for randomized mechanisms, as a first attempt of the problem we focus on a natural deterministic mechanism. 
In particular, we study the mechanism induced by \emph{approval voting}, which we call $\mechav$ with an appropriately selected reward function $\vecc f$. Specifically, the proposal with the highest amount of weighted approval is the winner, i.e., 
\begin{equation}
    x(\vecc r, \vecc w) \in \argmax_{j} \sum_{i \;:\; r_{ij} = 1} w_i \,.
\end{equation}
Ties can be broken arbitrarily, but in a deterministic manner, e.g., lexicographically. Hence, we might abuse the notation and use `$=$' with `$\argmax$'. It should be noted that while such naive tie-breaking rules are standard in theoretical work, in practice we expect to have a large number of experts at play with different and dynamically adjusting weights, hence a tie is very improbable anyway.   Additionally, the reward given to each expert $i$ is proportional to her weight:
\begin{equation}
    f_i({\vecc r}, \vecc w,  j^\star, q_{\winning}) = w_i \cdot 
        \begin{cases}
            a,&\text{if } r_{i \winning} = 1 \text{ and } q_{\winning} = 1 \\
            -s,&\text{if } r_{i \winning} = 1 \text{ and } q_{\winning} = 0\\
            a',&\text{if } r_{i \winning} = 0 \text{ and } q_{\winning} = 0 \\
            0,&\text{if } r_{i \winning} = 0 \text{ and } q_{\winning} = 1\\
        \end{cases}\,.
\end{equation}
That is, $a$ is the reward in case the expert approved the winning proposal and it turned out to have high quality, $a'$ is the reward in case the expert disapproved the winning proposal that turned out bad, $s$ is the penalty in case the expert approved the winning proposal and it turned out to be bad and we assume there is no reward or penalty  if the expert disapproved the winning proposal that turned out good. Notice that the collected reward depends on the winning proposal $\winning$ and it's quality. The other proposals are not implemented and their true nature is never revealed.

\begin{remark}
In the following analysis we will drop the $w_i$ multiplier. Indeed, all rewards are equally scaled, except for the $g_{ij}$ that do not depend on the weights. As such, to simplify notation (and without loss of generality) we consider that external rewards are scaled down appropriately by $w_i$ for each expert $i$.
\end{remark}

To find out the possible pure Nash equilibria of this scheme, we start by the simplest case for $\mechav$. Suppose everyone has already cast a vote and expert $i$ has no way of changing the outcome (which will usually be the most likely scenario). We need to check when the expected utility of approving is higher than that of disapproving:
\begin{align*}
    &a \cdot p_{i\winning} - s\cdot (1 - p_{i\winning}) + g_{i\winning} \cdot  p_{i\winning}
    \geq
    a' \cdot (1 - p_{i\winning}) + g_{i\winning} \cdot p_{i\winning} \\
    \Rightarrow\quad  & a \cdot p_{i\winning} - s + s \cdot p_{i\winning} \ge a' - a' \cdot  p_{i\winning}\\
    \Rightarrow\quad  &  p_{i\winning} \cdot (a + s + a') \ge a' + s\\
    \Rightarrow\quad  &  p_{i\winning} \ge \frac{a' + s}{a' + s + a}\,.
\end{align*}
So, the $a, a'$ and $s$ parameters can be tuned so that approving the winning proposal is the best option only for a confidence equal or higher than a desired threshold, which we define as
\begin{equation} \label{def:threshold}
    T = \frac{a' + s}{a' + s + a}\,.
\end{equation}
This threshold $T$ is used for measuring quality and allows us to define the `honest strategy' for this voting scheme.
\begin{definition}
Expert $i$ plays her \emph{honest strategy} if she approves the proposals for which her confidence is greater that $T$ and only those, i.e., if
$r_{ij} = 1 \Leftrightarrow p_{ij} \ge T$.
\end{definition}
\noindent
So far, we have shown that if every expert plays their honest strategy and it happens that no single expert has the power to change the outcome, then this honest strategy profile is a pure Nash equilibrium. The next result gives some insights about the possible deviations from the honest strategy, given the external rewards $g_{ij}$. Note that \Cref{thm:simple_pne} does not `protect' $\mechav$ against all possible deviations, but only those where an expert votes for a proposal she considers bad (amongst other things) to make it win. In the remaining possible deviations an expert determines the winning proposal, not by changing her vote for it, but by disapproving a proposal she considers good. As we shall see in Theorem \ref{thm:PoA}, such deviations do not hurt the overall quality significantly for semi-strategic experts. The intuition is that when the winning proposal is determined like this it necessarily has one of the highest number of honest votes and, thus, sufficiently high quality.


\begin{theorem}\label{thm:simple_pne}
    For any player $i$ and voting profile $\vecc r_{-i}$, let $\winning$ be the output of $\mechav$ if expert $i$ votes honestly and let $j'$ be any proposal that expert $i$ voted against. Then, if
    \begin{equation}\label{ineq:simple_pne_threshold}
    p_{ij'} < 
        \min\left\{
        T \cdot \frac{a + s}{a + s + g_{ij'}},
        \frac{a'\cdot(1 - T) + a}{a + s + g_{ij'}}
        \right\}\,,
    \end{equation}
    then expert $i$ cannot increase her payoff by switching her vote in favour of proposal $j'$ (and possibly also switching against $\winning$).
    Moreover, for any choice of parameters, there are instances where an expert $i$ can increase her payoff by switching her vote against proposal $\winning$.
\end{theorem}
\begin{proof}
Suppose that for some outcome $\vecc r$ the selected proposal is $\winning$, but expert $i$ could switch her vote to change the winner to another proposal $j'$ (either by not approving $\winning$ or by approving $j'$, and possibly approving / disapproving other proposals as well). For each possible deviation, we consider the expected utility of expert $i$ and show under which conditions switching the winner to $j'$ would be a better response, given the $g_{ij'}$ and $p_{ij'}$.
\begin{itemize}
    \item For the first group of cases, we assume that expert $i$ considers the winning proposal $\winning$ good enough (i.e., $p_{i\winning} \ge T$). Therefore, we need to compare the utility of any deviation that makes $j'$ the winner, to the utility obtained by voting `yes' to $\winning$. This is because voting `no' for $\winning$ is is clearly not a best response unless the winner changes, since $p_{i\winning} \ge T$ implies that $a \cdot p_{i\winning} - (1-p_{i\winning})\cdot s > a' \cdot (1-p_{i\winning})$.
    \begin{itemize}
    \item \textbf{Switches to approve $j'$, keeps approving $\winning$:} This deviation can only happen if the utility for switching is greater than that of voting honestly for $\winning$.
    \begin{align*}
        &a \cdot p_{ij'} - s \cdot (1 - p_{ij'}) + g_{ij'} \cdot p_{ij'}
        > 
        a \cdot p_{i\winning} - s \cdot (1 - p_{i\winning}) + g_{ij^\star}\cdot p_{ij^\star}\\
        \Rightarrow\quad  & a \cdot p_{ij'} - s \cdot (1 - p_{ij'}) + g_{ij'} \cdot p_{ij'}
        > 
        a \cdot p_{i\winning} - s \cdot (1 - p_{i\winning})\\
        \Rightarrow\quad  &  a \cdot p_{ij'} - s \cdot (1 - p_{ij'}) + g_{ij'} \cdot p_{ij'}
        >
        a \cdot T - s \cdot (1 - T)\\
        \Rightarrow\quad  &  a \cdot p_{ij'} - s + s \cdot p_{ij'} + g_{ij'} \cdot p_{ij'}
        >
        a \cdot T - s + s \cdot T\\
        \Rightarrow\quad  &  p_{ij'} \cdot (a + s + g_{ij'})
        >
        T \cdot (a + s)\\
        \Rightarrow\quad  &  p_{ij'} > T \cdot \frac{a + s}{a + s + g_{ij'}},
    \end{align*}
    using $p_{i\winning} > T$ in the first implication.
    
    \item \textbf{Switches to approve $j'$, switches to  disapprove $\winning$:} The incentives here are identical to the first case: in the honest outcome the reward is at least $a \cdot T - s \cdot (1 - T)$ and in the deviation it's $a \cdot p_{ij'} - s \cdot (1 - p_{ij'}) + g_{ij'} \cdot p_{ij'}$.

    \item \textbf{Keeps approving $j'$, switches to  disapprove $\winning$:}
    The incentives are identical to the first case, but since expert $i$ wanted to approve $j'$ (i.e., $p_{ij'} \ge T$), no matter how high $a, a'$ and $s$ are set she could better off disapproving $\winning$ in this scenario, if $p_{ij^\star} < p_{ij}$.
    \end{itemize}
 
    \item For the last two cases, we consider that $p_{i\winning} < T$. As before, any deviation that \emph{does not} change the winner to something other than $\winning$, needs contain a `no' vote for $\winning$: since $p_{i\winning} < T$ implies that $a' \cdot p_{i\winning} > a \cdot p_{i\winning} - (1-s)\cdot p_{i\winning}$.
    \begin{itemize}
    \item \textbf{Switches to approve $j'$, keeps disapproving $\winning$:} Since the honest move is to disapprove $\winning$, we have that $p_{i \winning} \le T$, therefore the `honest' payoff is at least $a' \cdot (1-T)$.
    \begin{align*}
        &a \cdot p_{ij'} - s \cdot (1 - p_{ij'}) + g_{ij'} \cdot p_{ij'}
        >
        a' \cdot (1 - p_{ij^\star}) + g_{ij^\star} \cdot p_{ij^\star}\\
        \Rightarrow\quad  & a \cdot p_{ij'} - s \cdot (1 - p_{ij'}) + g_{ij'} \cdot p_{ij'}
        >
        a' \cdot (1 - T)\\
        \Rightarrow\quad  & a \cdot p_{ij'} - s + s \cdot p_{ij'} + g_{ij'} \cdot p_{ij'}
        >
        a' \cdot (1 - T)\\
        \Rightarrow\quad  & p_{ij'} \cdot (a + s + g_{ij'})
        >
        a' \cdot (1 - T) + s\\
        \Rightarrow\quad  & p_{ij'}
        >
        \frac{a' \cdot (1 - T) + s}{a + s + g_{ij'}}
    \end{align*}
    
    \item \textbf{Keeps disapproving $j'$, switches to  disapprove $\winning$:}
    The rewards in the honest outcome are the same as in the first case, but the reward for deviating is different:
    \begin{align*}
        &a' \cdot (1 - p_{ij'}) + g_{ij'} \cdot p_{ij'}
        >
        a \cdot p_{ij^\star} - s \cdot (1 - p_{ij^\star}) + g_{ij^\star} \cdot p_{ij^\star}\\
        \Rightarrow\quad & a' \cdot (1 - p_{ij'}) + g_{ij'} \cdot p_{ij'}
        >
        a \cdot T - s \cdot (1 - T)\\
        \Rightarrow\quad  &  a' - a' \cdot p_{ij'} + g_{ij'} \cdot p_{ij'}
        >
        a \cdot T - s + s \cdot T \\
        \Rightarrow\quad  &  p_{ij'} \cdot (g_{ij'} - a')
        >
        a\cdot T - a' + s \cdot T - s
    \end{align*}
    If $g_{ij'} > a'$, then we need 
    \[
        p_{ij'}
        >
        \frac{a \cdot T - a' + s \cdot T - s}{g_{ij'} - a'}.
    \]
    Otherwise, we have that:
    \[
        p_{ij'}
        <
        \frac{a \cdot T - a' + s \cdot T - s}{g_{ij'} - a'}. 
    \]
    However, since $T > 0$ and the `no' branch of the reward function is decreasing in $p_{i\winning}$ we have that $a \cdot T - (1-T) \cdot s < a'$. Therefore, if $g_{ij'} > a'$, there is no way to set the other parameters and completely eliminate the possibility of deviating to disapproving $\winning$.
    \end{itemize}
\end{itemize}
    Putting everything together, this mechanism can only protect from situations where the expert needs to \emph{actively} switch her vote to approve a proposal she knows is not good enough.
    Following the previous cases, a necessary condition for this to happen is either:
    \[
    p_{ij'} > T \cdot \frac{a + s}{a + s + g_{ij'}}
    \]
    or
    \[
    p_{ij'} > \frac{a'\cdot (1 - T) + s}{a + s + g_{ij'}} \,.
    \]
    If $p_{ij'}$ is smaller than both, then there is no possibility of such a deviation and the claim holds.
\end{proof}

Setting $g_{ij} = 0$ yields the following corollary, showing that experts without any external rewards will never approve a perceived bad proposal.
\begin{corollary}
For any expert $i$, it is a dominant strategy to vote against any $j'$ such that $g_{ij'} = 0$ and $p_{ij'} < T$.
\end{corollary}
\begin{proof}
    If $g_{ij'} = 0$ (i.e., the expert in question has no external motivations), the previous probabilities take on interesting values. Specifically, there is no deviation for:
    \[
        p_{ij'} < T \cdot \frac{a + s}{a + s} = T 
    \]
and
    \[
        p_{ij'} < \frac{a'\cdot (1 - T) + s}{a + s} \cdot T =
        \frac{a \cdot T - (1-T) \cdot s + s}{a + s} \cdot T < T \,.
    \]
\end{proof}
The previous two results show that, assuming the $a, a'$ and $s$ are all large enough, the experts will be reluctant to vote in favour of a proposal they already know is bad: to do so, they still need to have some faith in that proposal. It turns out that this mechanism always has pure Nash equilibria, albeit with limited guarantees, as shown in the next two propositions.

\begin{proposition}\label{lemma:existence_approval}
    In the presence of strategic experts, the approval voting mechanism $\mechav$ always has a PNE.
\end{proposition}
\begin{proof}
    For $i\in N$ and $j\in U$, let $\vecc r^{ij}$ be the voting profile where
    \begin{itemize}
        \item $r_{i j} = 1$
        \item $r_{i'j'} = 0$ for any $i'$ and $j' \neq j.$
    \end{itemize}
    If no expert $i$ has positive utility for some proposal $j$ with respect to the voting profile $\vecc r^{ij}$, then clearly everyone voting `no' to every proposal is a pure Nash equilibrium. 
 
    Assuming that the set $N_+ = \{i\in N \;|\; \exists j\in U : u_{i}^{\mechav}(\vecc r^{i j}\;|\; \vecc w, \vecc p) > 0\}$ is nonempty, let $i^\star = \argmax_{i\in N_+} w_i$. Additionally, let $j^\star = \argmax_{j\in U} u_{i^\star}^{\mechav}(\vecc r^{i^\star j}\;|\; \vecc w, \vecc p)$. It is not hard to see that the profile $\vecc r^{i^\star j^\star}$ is a pure Nash equilibrium. 
    Since the weight of every  expert $i \neq i^\star$ is $w_i < w_{i^\star}$ (or $w_i = w_{i^\star}$ but $i$ is losing to $i^\star$ in the tie-breaking), there exists no possible deviation from $i$ that changes the winning proposal. In addition, because $j^\star$ maximizes the utility of expert $i^\star$, this is a pure Nash equilibrium.
\end{proof}

\begin{proposition}\label{prop:PoA_strategic}
    The Price of Anarchy of $\mechav$ is $\Omega(n)$, even if for all experts $i$ and all proposals $j$ we have that $g_{ij} = 0$.
\end{proposition}
\begin{proof}
    Suppose that there are $n+1$ experts and $2$ proposals. For arbitrarily small $\epsilon > 0$, we set:
    \begin{itemize}
        \item For expert $1$: $w_1 = 1/n + \epsilon$, $p_{11} = 1$ and $p_{12} = 0$.
        \item For any expert $i > 1$: $w_i = 1/n$, $p_{i1} = 0$ and $p_{i2} = 1$.
    \end{itemize}
    Following the construction of \Cref{lemma:existence_approval}, there is a pure Nash equilibrium where expert $1$ approves the first proposal and every other expert votes `no' for all proposals. The quality of proposal 1 is $1/n + \epsilon$, while the quality of proposal $2$ is $n\cdot 1/n = 1$, leading to the claimed result.
\end{proof}

This equilibrium of \Cref{prop:PoA_strategic}, however, is unnatural: why would so many experts vote against their favourite proposal? The intuition is that the assumption about the agents being semi-strategic instead, should help us avoid such pitfalls.
Unfortunately, if we assume the presence of semi-strategic agents, there are combinations of $p_{ij}$'s for which no PNE exists.
\begin{proposition}\label{thm:no_pne}
    The mechanism $\mechav$ does not always have PNE for semi-strategic experts, even when $g_{ij} = 0$ for all experts $i$ and proposals $j$.
\end{proposition}
\begin{proof}
     Suppose that we have 3 experts and 2 proposals, $T=0.9$ (with $a, a'$ and $s$ set appropriately), and
    \begin{itemize}
        \item expert $1$ has $w_1 = 0.49$, $p_{11} = 0.95, p_{12} = 1$;
        \item expert $2$ has $w_2 = 0.41$, $p_{21} = 1, p_{22} = 0.95$;
        \item expert $3$ has $w_3 = 0.1$, $p_{31} = 1, p_{32} = 0$.
    \end{itemize}
    Clearly, expert $3$ would always vote for proposal 1 only. The remaining experts would honestly approve both proposals, but expert 1 has higher expected reward if proposal 2 wins, while expert 2 believes proposal 1 maximizes her utility. We present the following cycle of deviations. At first they all vote honestly. Then expert 1 says `no' to the first proposal, making proposal 2 the winner and improving her utility. In response, expert 2 says rejects the second proposal, leading proposal 1 to reclaim the win. Then, expert `1' has to switch her vote in favour of proposal 1 because they are semi-strategic. Finally, expert 2 says `yes' to the second proposal because they are semi-strategic too.
    
    Therefore, at any configuration of votes, some expert wants to deviate and $\mechav$ has no PNE for this instance.
\end{proof}

Despite \Cref{thm:no_pne}, our mechanism does have approximate PNE's for semi-strategic experts and an appropriate choice of parameters, as we show next. Moreover, these approximate equilibria always lead to choosing approximately optimal proposals.

\subsection{Approximate Equilibria of $\mechav$}
\label{sec:approx-equilibria}
Since pure Nash equilibria may not always exist when dealing with semi-strategic experts, we have to fall back to showing the existence of approximate PNE's. This can be achieved by careful tuning of the parameters $a, a'$ and $s$ when defining the reward function. Recall that $a, a', s$ and $T$ are related via Equation \eqref{def:threshold}; an equivalent equation appears in the proof of the theorem below as \eqref{eq:inflection}.
\begin{theorem}\label{thm:PNE}
    Suppose that for $T\in[0,1]$, we set $\epsilon \ge 0$ such that:
    \begin{itemize}
        \item $1/(\epsilon + 1) < T$
        \item $a = (1+\epsilon)\cdot a'\cdot (1-T) > a'$
        \item $s = \frac{a \cdot \left(T \cdot (\epsilon + 1) - 1\right)}{(1 - T)\cdot(\epsilon + 1)}$
    \end{itemize}
    In addition, suppose that for every player $i$ and proposal $j$ we have that $g_{ij} \le a \cdot \delta$. Then the voting profile $\vecc r$ where everyone votes honestly is $(1+\epsilon)\cdot(1+\delta)$-approximate pure Nash equilibrium.
\end{theorem}
\begin{proof}
As always, the honest behaviour of expert $i$ should be to approve proposal $j$ if and only if they have $p_{ij} \ge T$. From the perspective of expert $i$, their reward for voting in favour of proposal $j$ is:
\[
p_{ij}\cdot a - (1-p_{ij}) \cdot s.
\]
Notice that this expression is strictly increasing in $p_{ij}$.
Voting against proposal $j$ yields an expected reward equal to
$a' \cdot (1-p_{ij})$, which is strictly decreasing in $p_{ij}$. As before, to ensure honest behaviour, the two expressions need to be equal for $p_{ij} = T$:
\begin{equation}\label{eq:inflection}
    \quad T \cdot a - (1 - T)\cdot s = a' \cdot (1-T) \,.
\end{equation}
Additionally, we need that $a$, which is the payoff for $p_{ij} = 1$ is also the maximum possible reward and satisfies $a = (1+\epsilon)\cdot (a\cdot T - (1-T)\cdot s)$. Since $T \in [0,1]$, \Cref{eq:inflection} is actually the global minimum of the honest response reward. Therefore, the condition that $a$ is the maximum can be replaced by $a \ge a'$, since either $a$ or $a'$ are the extreme points. Since $T$ and $\epsilon$ are given, we can solve for the remaining values using $a' \ge 0$ as a the free parameter. Moreover, these solutions are non-negative for $1/(\epsilon + 1)<T<1$.

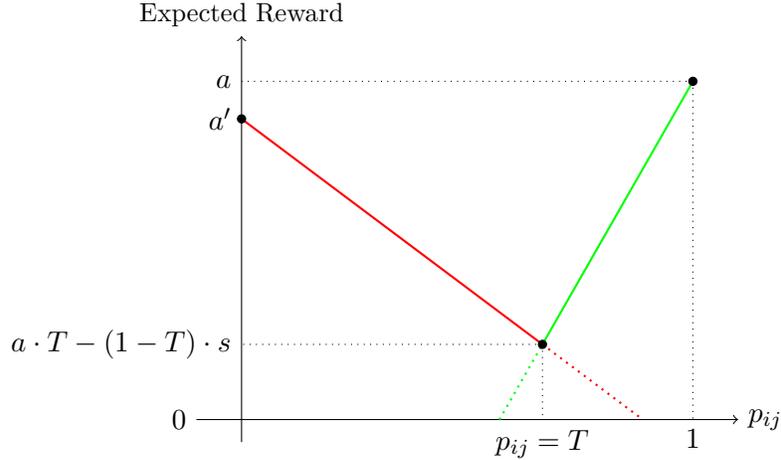
\begin{figure}[h]
    \centering
    \begin{tikzpicture}[scale=3, xscale = 2, yscale=1, domain=0:1, range=0:1, samples=400, smooth]
    
    \draw[->] (-0.1, 0) node[left] {$0$} -- (1.1, 0) node[right] {$p_{ij}$};
    
    \draw[->] (0, -0.1) -- (0, 1.7) node[above] {\small Expected Reward};
    
    \draw[color=red, thick, domain=0:0.6666] plot (\x,{-1.5 * \x + 4/3});
    \draw[dotted, color=red, thick,domain=0.6666:8/9] plot (\x,{-1.5 * \x + 4/3});

    \draw[dotted, color=green, thick, domain=0.5714:0.6666] plot (\x,{1.5 * \x - (1 - \x) * 2});
    \draw[color=green, thick, domain=0.6666:1] plot (\x,{1.5 * \x - (1 - \x) * 2});

    \draw[dotted] (0.66666, {0.3333333}) -- (0.6666666666, 0) node[below] {$p_{ij} = T$};
    \node [scale=0.3, draw, circle, fill] at (0.66666, 0.333333) {};
    
    \draw[dotted] (0.66666, {0.3333333}) -- (0, 0.333333) node[left] {$\quad a \cdot T - (1-T)\cdot s$};
    
    \node [scale=1, left] at (0, 4/3) {$a'$};
    \node [scale=0.3, draw, circle, fill] at (1, 1.5) {};
    \node [scale=0.3, draw, circle, fill] at (0, 4/3) {};   
    
    \draw[dotted] (1, 1.5) -- (0, 1.5) node[left] {$a$};
    \draw[dotted] (1, 1.5) -- (1, 0) node[below] {1};
    
    \end{tikzpicture}
    
    \caption{The expected reward (as a function of $p_{ij}$) for honest voting, assuming $g_{ij} = 0$ and proposal $j$ won. The red line is $a' \cdot (1 - p_{ij})$, corresponding to voting against $j$, while the green line is $a \cdot p_{ij} - (1-p_{ij}) \cdot s$ for voting in favour of it.}
    \label{fig:function_example}
\end{figure}

We are now ready to show that $\vecc r$, the honest voting profile, is an approximate PNE.
Let $\winning$ be the winning proposal. Clearly, any expert who \emph{cannot} change the outcome is playing their best response. Suppose that expert $i$ has a beneficial deviation $\vecc r_i'$ and changes the winner to $j'$. There are two cases:
\begin{itemize}
    \item The expert $i$ is honest about $j'$ (but possibly changed his vote on some other proposals): this means that either $p_{ij'} \ge T$ and $r_{ij'} = 1$ or $p_{ij'} < T$ and $r_{ij'} = 0$. In this case, the maximum possible reward she could get is $a + g_{ij'} \le (1 + \delta)\cdot a$. On the other hand, the minimum possible reward for an honest vote is $a'\cdot(1-T) = a / (1 + \epsilon)$. Therefore, this deviation can yield at most $(1+\epsilon)\cdot(1 + \delta)$ times the reward of the honest response.
    \item The expert $i$ is dishonest about $j'$: in this case, the reward (without $g_{ij}$) is at most $a' \cdot (1-T)$, which is the minimum possible reward for honest voting. As with the previous case, the addition of $g_{ij'}$ is not great enough to motivate the expert to deviate.
\end{itemize}

Therefore, the honest profile is an ($1+\epsilon$)-pure Nash equilibrium.
\end{proof}

Note that the existence of the approximate PNE of \Cref{thm:PNE} is not guaranteed for an arbitrarily small $\epsilon$. So, it is natural to ask how inefficient these equilibria are, with respect to achieving our objective of maximizing $\texttt{Qual}$. We deal with this question in the following section.

\subsection{Price of Anarchy of $\mechav$}
\label{sec:PoA}
\newcommand{\winnerPNE}{j}
\newcommand{\winnerOPT}{j^\star}

Here we study the Price of Anarchy of the approximate pure Nash equilibria of $\mechav$. That is, we bound the quality of a proposal returned by the mechanism in an approximate equilibrium in terms of the best possible estimated quality. Surprisingly, we show that for any $\epsilon$, $(1+\epsilon)$-approximate PNE result in quality which is within a factor of $2$ of the optimal estimated quality. Note that although the statement of \Cref{thm:PoA} does not mention the $g_{ij}$'s explicitly, these are taken into consideration via the conditions of \Cref{thm:simple_pne}. Moreover, this bound on the Price of Anarchy is tight.

\begin{theorem}\label{thm:PoA}
    Suppose that $a,a'$ and $s$ are chosen such that:
    \begin{itemize}
        \item $a = (1 + \epsilon) \cdot (1-T) \cdot a'$.
        \item $(1-T)\cdot a' = T\cdot a - (1-T)\cdot s$.
    \end{itemize}
    In addition, the $p_{ij} < T$ of every expert satisfy the conditions of \Cref{thm:simple_pne}. Then, the Price of Anarchy of $\mechav$ over $(1+\epsilon)$-approximate pure Nash Equilibria is at most $2$.
\end{theorem}
\begin{proof}
    Let $\vecc r$ be a $(1+\epsilon)$ approximate PNE whose winner is $\winnerPNE$. Further, let $\winnerOPT$ be the proposal with highest quality and suppose that:
    \[
        \texttt{Qual}[\winnerPNE] < \frac{1}{2}\cdot \texttt{Qual}[\winnerOPT].
    \]

Each expert $i$ belongs to one of the following categories:
\begin{itemize}
    \item \textbf{Case 1}: $p_{i\winnerPNE} < T$ and $p_{i\winnerOPT} < T$: In this case, the expert has to disapprove both proposals at the equilibrium $\vecc r$. By \Cref{thm:simple_pne}, expert $i$ would gain no benefit by voting in favour of either $j$ or $j^\star$. 
    \item Either $p_{i\winnerPNE} \ge T$ or $p_{i\winnerOPT} \ge T$:
    \begin{itemize}
        \item \textbf{Case 2a}: In the first case, if expert $i$ submits a `no' vote for $\winnerPNE$ and it remains the winner, this `no' vote also clearly reduces the reward of expert $i$ compared to a `yes' vote. Since expert $i$ is semi-strategic, they have to vote in favour of $\winnerPNE$.
        \item \textbf{Case 2b}: In the second case, by \Cref{thm:simple_pne} they cannot approve $\winnerPNE$. Since $\vecc r$ is an equilibrium where $\winnerPNE$ wins and they are semi-strategic, voting `no' for $\winnerPNE$ does not strictly increase their reward. Therefore expert $i$ votes only in favour of $\winnerOPT$.
    \end{itemize}
    \item \textbf{Case 3}: $p_{i\winnerPNE} \ge T$ and $p_{i\winnerOPT} \ge T$: Similarly to \textbf{Case 2}, the expert has to approve $\winnerPNE$. However, not every expert needs to vote in favour of $\winnerOPT$. They only do so if permitted by the equilibrium condition (i.e., if the winner stays $\winnerPNE$).  
\end{itemize}
    We partition the experts into sets $C_1, C_{2a}, C_{2b}$ and $C_3$ respectively, indexed according to the aforementioned cases. In addition, let $\rho \in [0,1]$ be the fraction of experts in $C_3$ that voted for $\winnerOPT$ as well as $\winnerPNE$. Clearly, since proposal $\winnerPNE$ is the winner we have that:
    \begin{equation}\label{eq:part_a}
        \sum_{i \in C_{2a}} w_i + \sum_{i \in C_{3}} w_i \ge \sum_{i \in C_{2b}} w_i + \rho \cdot \sum_{i \in C_{3}} w_i
        \Rightarrow
        \sum_{i \in C_{2a}} w_i + (1-\rho) \cdot \sum_{i \in C_{3}} w_i \ge \sum_{i \in C_{2b}} w_i.
    \end{equation}
    In addition, by since $\winnerOPT$ maximizes the quality objective, we have:
    \begin{equation}\label{eq:part_b}
        \sum_{i \in C_{2a}} w_i + \sum_{i \in C_{3}} w_i < \frac{1}{2} \cdot \left(\sum_{i \in C_{2b}} w_i + \sum_{i \in C_{3}} w_i\right)
        \Rightarrow  2 \cdot \sum_{i \in C_{2a}} w_i + \sum_{i \in C_{3}} w_i < \sum_{i \in C_{2b}} w_i.
    \end{equation}
    Combining \Cref{eq:part_a} with \Cref{eq:part_b} we get:
    \begin{equation*}
        \sum_{i \in C_{2a}} w_i + (1-\rho) \cdot \sum_{i \in C_{3}} w_i >  \sum_{i \in C_{2a}} w_i + \sum_{i \in C_{3}} w_i,
    \end{equation*}
    which is impossible for any $\rho \in [0,1]$, leading to a contradiction.
\end{proof}
We complement the previous theorem with a matching lower bound.
\begin{theorem}
    The Price of Anarchy of $\mechav$ is greater than or equal to 2.
\end{theorem}
\begin{proof}
Consider and instance with 2 experts and 2 proposals, with the following parameters:
\begin{itemize}
    \item Expert 1 has $p_{11} = T < p_{12} = 1$ and $w_1 = 1 + \epsilon$.
    \item Expert 2 has $p_{11} = 1, p_{12} = 0$ and $w_2 = 1 - \epsilon$.
\end{itemize}
All $g_{ij}$ are equal to zero.\smallskip

The optimal outcome is to elect proposal 1, that has quality $2$. However, it is a semi-strategic deviation for expert 1 to vote \emph{against} proposal 1, since she likes proposal 2 slightly more, even though both meet the acceptance threshold $T$. In this case, expert 2 has no way to change the outcome with her lower weight, leading to an exact PNE with quality $1 + \epsilon$.
\end{proof}

\section{The Repeated Game}
\label{sec:repeated}
In the previous sections we described a system which incentivizes the experts to only vote `yes' for proposals that they believe have a high chance of being good. While we have defined the reward of each expert to be proportional to her weight, this does not have any significant impact on our technical results so far. As mentioned in the introduction, however, we want the weight of an expert to serve as a proxy for that expert's demonstrated expertise level, capturing her ``reputation'' in the system. This, of course, makes sense in a repeated game setting, where the weights are updated after each round of proposals. We assume that every time we have a fresh set of proposals, independent of any past decisions, but the different parameters of the system (threshold $T$, reward parameters $a$, $a'$ and $s$, etc.) remain the same and there is a known rule for updating the weights.

An analog of the various Folk Theorems (see, e.g., \cite{Friedman71,FudenbergMaskin86}) would not apply in our setting with the semi-strategic experts, since the notion of a ``threat'' used in their proofs cannot be used anymore. Nevertheless, we  show below that if rewards are smoothed out appropriately, then truth-telling is an approximate pure Nash equilibrium. This, combined with Theorems \ref{thm:PNE} and \ref{thm:PoA} directly gives us a \emph{Price of Stability} (which is the ratio between the quality of the best possible equilibrium and the optimal outcome) of $2$ for this repeated game.

Let $\vecc w^t = (w_1^t, w_2^t, \ldots, w_n^t)$ be the weights after round $t\ge 1$. In this context, a voting mechanism will involve two components: a reward function $\vecc f^t(\vecc r, \vecc w^t, j, q_j)$ and a weight update rule $\vecc w^{t+1} = g(\vecc w^t, \vecc r, j, q_j)$. The reward function will be the same to the single-shot game: $f_i^t({\vecc r}, \vecc w^t,  j, q_{j}) = f_i({\vecc r}, \vecc w^t, j, q_{j})$. 
For the sake of presentation, we will focus on a simple weight update rule here, so that the weights converge to the percentage of correct predictions; the same argument, however, could be made for \emph{any} update rule. 

In principle, we would not like the weights to fluctuate widely from round to round, since then they would not capture the empirical expertise level as intended. Suppose we define $$\omega_i^t = \frac{\# \text{ of correct predictions}}{t}$$. Even if we assume that each expert has an inherent expertise level $\pi_i$ so that $\lim_{t\to\infty} \omega_i^t = \pi_i$, these weights can still fluctuate a lot when $t$ is small. Having these weights as a starting point, however, for a small $\zeta>0$, we may define $\vecc w^t$ as follows: 
\begin{align*}
w_i^{0}&=1/2\,; \\
w_i^{t+1}&= 
        \begin{cases}
            \min\{ \omega_i^t, (1+\zeta) \cdot w_i^t\}, &\text{if } t\ge 1 \text{ and } w_i^t \le \omega_i^t \\
            \max\{ \omega_i^t, (1-\zeta) \cdot w_i^t\}, &\text{if } t\ge 1 \text{ and } w_i^t > \omega_i^t \\
        \end{cases}
\end{align*}
Using these ``delayed updates'', we still have $\lim_{t\to\infty} w_i^t = \pi_i$, but the weights never change more than $100\cdot\zeta\,\%$ from round to round. Note that is is not necessary to start with a rule that converges in any sense.

As usual, we  assume that future rewards are discounted by a \emph{discount factor} $\gamma\in(0,1)$. That is, an amount of money $x$ that is expected to be won $\tau$ rounds into the future, has value $\gamma^{\tau}x$ at the present moment for any of the experts.

\begin{theorem}\label{thm:PNE-repeated}
    Let $\xi\in (0,1)$. Also let $T$, $\epsilon$, $\delta$, $a$, $a'$, and $s$ be like in Theorem \ref{thm:PNE},
    and suppose that for every player $i$ and any proposal $j$ of any round $t$, we have $g_{ij} \le w^t_i \cdot a \cdot \delta$. 
    Then the sequence of voting profiles $(\vecc r^t)_{t\in\mathbb{N}}$ where everyone votes honestly in each round $t$ is $(1+3\epsilon)\cdot(1+\delta)$-approximate pure Nash equilibrium for the Repeated Update Selection game with delayed weight updates and sufficiently small discount factor $\gamma$.
\end{theorem}
\begin{proof}
Fix any expert $i$. Let $w_i^t$ be the weights of  $i$ for $t\in \mathbb{N}$ if $(\vecc r^t)_{t\in\mathbb{N}}$ is played. Then, in round $t$, the expected reward of $i$ (before the discount) is at least $w_i^t \cdot (1-T) \cdot a'$.
Now, consider any sequence of voting profiles $(\tilde{r}^t_i, \vecc r^t_{-1})_{t\in\mathbb{N}}$, and let $\tilde{w}_i^t$ be the weights of  $i$ if that sequence is played. Now, in round $t$, the expected reward of $i$ (before the discount) is at most $(1+\delta) \cdot \tilde{w}_i^t \cdot a$. When calculating the expected overall rewards $R((\vecc r^t)_{t\in\mathbb{N}})$ and $R((\tilde{r}^t_i, \vecc r^t_{-1})_{t\in\mathbb{N}})$, we have
\begin{align*}
R((\vecc r^t)_{t\in\mathbb{N}}) &\ge \sum_{t=0}^{\infty} [ w_i^t \cdot (1-T) \cdot a' \cdot \gamma^t ]\\ &\ge \sum_{t=0}^{\infty} [ w_i^0 \cdot (1-\zeta)^t \cdot (1-T) \cdot a' \cdot \gamma^t ] \\
&= w_i^0 \cdot (1-T) \cdot a' \cdot \sum_{t=0}^{\infty} [(1-\zeta) \cdot \gamma]^t \\ &=  \frac{w_i^0 \cdot(1-T) \cdot a'}{1-(1-\zeta) \cdot \gamma}, 
\end{align*}
as well as
\begin{align*}
R((\tilde{r}^t_i, \vecc r^t_{-1})_{t\in\mathbb{N}}) &\le \sum_{t=0}^{\infty} [ (1+\delta) \cdot \tilde{w}_i^t \cdot a \cdot \gamma^t]\\
&\le \sum_{t=0}^{\infty} [ (1+\delta) \cdot w_i^0 \cdot (1+\zeta)^t \cdot (1+\varepsilon) \cdot (1-T) \cdot a' \cdot \gamma^t ]\\
&= (1+\delta)\cdot (1+\varepsilon) \cdot w_i^0 \cdot (1-T) \cdot a' \cdot \sum_{t=0}^{\infty} [(1+\zeta) \cdot \gamma]^t\\ 
&=  (1+\delta)\cdot(1+\varepsilon) \cdot \frac{w_i^0 \cdot (1-T) \cdot a'}{1-(1+\zeta) \cdot \gamma} \,.
\end{align*}
It is a matter of simple calculation to see that for small enough $\gamma$, $\frac{1-(1-\zeta)\cdot\gamma}{1-(1+\zeta)\cdot\gamma}\le (1+\varepsilon)$, and thus, 
\begin{align*}
R((\tilde{r}^t_i, \vecc r^t_{-1})_{t\in\mathbb{N}}) &\le (1+\delta) \cdot (1+\varepsilon)^2 \cdot R((\vecc r^t)_{t\in\mathbb{N}}) \\ &\le (1+\delta) \cdot (1+3\varepsilon) \cdot R((\vecc r^t)_{t\in\mathbb{N}}) \,,  
\end{align*}
as desired.
\end{proof}

Given that repeated games introduce a large number of (approximate) pure Nash equilibria, some of which may be of very low quality, it is not possible to replicate our Price of Anarchy result from Theorem \ref{thm:PoA} here. However, that result, coupled with the fact that the sequence of voting profiles in Theorem \ref{thm:PNE-repeated} consists of approximate pure Nash equilibria of the single-shot game (Theorem \ref{thm:PNE}),
directly translate into the following Price of Stability result.

\begin{corollary}\label{cor:PoS-repeated}
Under the assumptions of Theorems \ref{thm:PNE-repeated} and \ref{thm:PoA} (where the original $g_{ij}$ is replaced by $g_{ij}/w^t_i$ if $j$ is an update of round $t$), the Repeated Update Selection game with delayed weight updates and sufficiently small discount factor $\gamma$ has Price of Stability at most $2$.
\end{corollary}

\section{Conclusions and Open Problems}
In this work, we make the first step in combining aspects of voting and eliciting truthful beliefs at once. Typically voting applications do not involve payments dependent on votes and the chosen outcome, whereas information elicitation (using peer prediction or proper scoring rules) do not involve the experts affecting the chosen action. Using the notion of semi-strategic experts and a variant of approval voting, with appropriate rewards, we can prove the existence of approximate pure Nash equilibria and show that they produce outcomes of good quality. This is only the first step however. The natural follow-up questions would be:
\begin{itemize}
    \item To study randomized mechanisms, which might have better guarantees. For instance, rather than always selecting the proposal with highest approval, the selection could be randomized between the top proposals that exceed a threshold. In this case, experts might be more cautious about lying, as they have less influence on the final selection. In the deterministic case they can be sure that their deviation yielded some benefit. But with randomization, there is always the chance that a different proposal is chosen and it might not be dishonest with too many votes.
    \item To allow experts a richer strategy space, including reporting their prior (i.e., the $p_{ij}$) directly, or even estimating how other experts might act.
\end{itemize}
Either of these approaches could produce a mechanism with improved performance, at a cost of added complexity and perhaps lower robustness.

\section{Acknowledgements}
We would like to thank Nikos Karagiannidis for many enlightening meetings, helping us formulate the model in early versions of this work.

\bibliography{main}
\bibliographystyle{plainnat}

\end{document}